\documentclass[11 pt, letterpaper, english]{article}
\usepackage[utf8]{inputenc}
\usepackage[letterpaper, top=1in, bottom=1in, left=1in, right=1in]{geometry}

\usepackage{mathtools}
\usepackage{calligra}
\usepackage{lmodern}
\usepackage[T1]{fontenc}
\usepackage{anysize}
\usepackage{amssymb}
\usepackage[english]{babel}
\usepackage{soul} 
\usepackage{float}
\usepackage{enumerate}
\usepackage{color}
\usepackage{comment}
\usepackage{amsmath}
\usepackage{tikz}
\usepackage[ruled,vlined]{algorithm2e}

\usepackage{amsthm}
\usepackage{authblk}
\usepackage{multirow}
\usepackage{booktabs}
\usepackage{thmtools}
\usepackage{dsfont}

\usepackage{hyperref}
\usepackage[capitalise, noabbrev]{cleveref}

\theoremstyle{plain}
\newtheorem{thm}{Theorem}
\newtheorem{lem}[thm]{Lemma}

\theoremstyle{definition}

\theoremstyle{remark}

\newtheorem{obs}[thm]{Observation}

\crefname{lem}{Lemma}{lemmas}

\usepackage{todonotes}

\newcommand{\ALG}{\textsc{ALG}}

\newcommand{\PP}{\mathbb{P}}
\newcommand{\E}{\mathbb{E}}

\newcommand{\ALGO}{\ensuremath{\mathrm{ALG_o}} }
\newcommand{\ALGC}{\ensuremath{\mathrm{ALG_c}} }
\newcommand{\ALGF}{\ensuremath{\mathrm{ALG_f}} }

\newcommand{\UU}{\mathcal{U}}
\newcommand{\DD}{\mathcal{D}}

\title{The Two-Sided Game of Googol and Sample-Based Prophet Inequalities}
\author{Jos\'e Correa}
\author{Andr\'es Cristi}
\author{Boris Epstein }
\author{Jos\'e A. Soto}
\affil{Universidad de Chile}

\date{}

\begin{document}

\maketitle

\begin{abstract}
The secretary problem or the game of Googol are classic models for online selection problems that have received significant attention in the last five decades. In this paper we consider a variant of the problem and explore its connections to data-driven online selection. Specifically, we are given $n$ cards with arbitrary non-negative numbers written on both sides. The cards are randomly placed on $n$ consecutive positions on a table, and for each card, the visible side is also selected at random. The player sees the visible side of all cards and wants to select the card with the maximum hidden value. To this end, the player flips the first card, sees its hidden value and decides whether to pick it or drop it and continue with the next card.

We study algorithms for two natural objectives. In the first one, similar to the secretary problem, the player wants to maximize the probability of selecting the maximum hidden value. We show that this can be done with probability at least $0.45292$. In the second objective, similar to the prophet inequality, the player wants to maximize the expectation of the selected hidden value. Here we show a guarantee of at least $0.63518$ with respect to the expected maximum hidden value.

Our algorithms result from combining three basic strategies. One is to stop whenever we see a value larger than the initial $n$ visible numbers. The second one is to stop the first time the last flipped card's value is the largest of the currently $n$ visible numbers in the table. And the third one is similar to the latter but to stop it additionally requires that the last flipped value is larger than the value on the other side of its card.

We apply our results to the prophet secretary problem with unknown distributions, but with access to a single sample from each distribution. In particular, our guarantee improves upon $1-1/e$ for this problem, which is the currently best known guarantee and only works for the i.i.d. prophet inequality with samples.
\end{abstract}

\section{Introduction}

In the classic game of Googol we are given $n$ cards with $n$ different arbitrary positive numbers written on them. The cards are shuffled and spread on a table with the numbers facing down. 
The cards are flipped one at a time, in a random uniform order, and we have to decide when to stop. The goal is to maximize the probability that the last flipped card has the overall greatest number.

In this paper we study a variant of this problem that we call {\em The two-sided game of Googol}.  
Similar to the classic version, we are given $n$ cards that we have to flip in random uniform order. However, here the cards have numbers on both sides, so we have $2n$ different arbitrary non-negative numbers instead of $n$, written on each side of each card. The cards are shuffled and spread
on a table so that, independently for each card, either side faces up with probability $1/2$. We can see all the numbers that landed facing up (while the other side is hidden), and flip one card at a time, revealing the number that was facing down. Again we have to decide when to stop. We study the \emph{secretary} and the \emph{prophet} variants. In the secretary variant, the goal is to maximize the probability of stopping at the maximum number over the numbers that landed facing down. In the prophet variant, the goal is to maximize the ratio between the expectation of the last number revealed before stopping, and the expected maximum of the numbers that landed facing down.

This problem naturally fits within the theory of optimal stopping theory which is concerned with choosing the right time to take a particular action, so as to maximize the expected reward. Two landmark examples within this theory are the secretary problem (or game of Googol) described above, and the prophet inequality. In the latter  a gambler faces a finite sequence of non-negative independent random variables $X_1, \ldots, X_n$ with known distributions $F_i$ from which iteratively a prize is drawn. After seeing a prize, the gambler can either accept the prize and the game ends, or reject the prize and the next prize is presented to her. The classical result of Krengel and Sucheston, and Garling \cite{KS77,KS78}, states that the gambler can obtain at least half of the expected reward that a prophet can make who knows the realizations of the prizes beforehand. That is, $\sup \{ \E[X_{t}] \, : \, t \text{ stopping rule} \, \} \geq \frac{1}{2} \E\{\sup_{1 \leq i \leq n} X_i\}$. Moreover, Krengel and Sucheston also showed that this bound is best possible. Remarkably, Samuel-Cahn \cite{SC84} showed that the bound of $1/2$ can be obtained by a simple threshold rule, which stops as soon as a prize is above a fixed threshold. We refer the reader to the survey by Hill and Kertz \cite{HK92} for further classic results.

In recent years, motivated by the connections between optimal stopping and posted price mechanisms \cite{HKS2007,Chawla2010,CFPV19}, there has been a regained interest in studying algorithms for variants of the classic prophet inequalities. The recent survey of Lucier \cite{L2017} is a good starting point to understand the economic perspective of prophet inequalities, while the recent letter of Correa et al. \cite{CFHOV18} provides a recent overview of results for single choice prophet inequalities. Due to this regained interest, new variants of the prophet inequality setting have been studied, including problems when the gambler has to select multiple items, when the selected set has to satisfy combinatorial constraints, when the market is large, when prior information is inaccurate, among many others (see e.g. \cite{AE17,KW12,DFKL17,EHKS18,DK19}). 

Particularly relevant to our work in the \emph{prophet secretary} problem. In this version, the random variables are shown to the gambler in uniform random order, as in the secretary problem. The problem was first studied by Esfandiari et al. \cite{esfandiari2015} who found a bound of $1-1/e$. Later, Eshani et al. \cite{EHKS18} showed that the bound of $1-1/e$ can even be achieved using a single threshold. The factor $1-1/e$ was first beaten by Azar et al \cite{ACK18} by a tiny margin, while the current best bound is 0.67 \cite{CSZ19}. In terms of upper bounds it was unclear until very recently whether there was a separation between prophet secretary and the i.i.d prophet inequality, where the random variables are identically distributed. For this problem, Hill ad Kertz \cite{HK82} provided the family of worst possible instances from which \cite{K86} proved the largest possible bound one could expect is $0.7451$ and Correa et al. \cite{CFHOV17} proved that this value is actually tight. Very recently, Correa et al. show that no algorithm can achieve an approximation factor better than $\sqrt{3}-1\approx 0.71$ for the prophet secretary problem (see the full version of \cite{CSZ19}). Interestingly, the tight is still~unknown.

On the other hand, some recent work has started to investigate data-driven versions of the prophet inequality since the full distributional knowledge seems quite strong for many applications. In this context, Azar et al. \cite{AKW14} consider a version in which the gambler only has access to one sample from each distribution. They prove a prophet inequality with an approximation guarantee of 1/4 for this problem and left open whether achieving 1/2 is possible. This question was recently answered on the positive by Wang \cite{W18} who uses an elegant approach to prove that just taking the maximum of the samples as threshold leads to the optimal guarantee. If the instance is further specialized to an i.i.d.~instance, the best known bound is $1-1/e$ while it is known that one cannot achieve a guarantee better than $\ln(2)\approx 0.69$ \cite{CDFS19}.

Our problem is closely related to the single-sample version of the secretary problem and the prophet secretary problem, which combine the data-driven approach of Azar et al \cite{AKW14} and the random arrival order of Esfandiari et al. \cite{esfandiari2015}. In these problems we are given $n$ distributions, which are unknown to us, and only have access to a single sample from each. Then $n$ values are drawn, one from each distribution, and presented to us in random order. When we get to see a value, we have to decide whether to keep it and stop, or to drop it and continue. Again in the  single-sample secretary problem the goal is to maximize the probability of stopping with the largest value while in the single-sample prophet secretary problem the goal is to maximize the expectation of the value at which we stopped divided by the expectation of the maximum of the values. 

It is immediate to observe that an algorithm for the two-sided game of Googol that has a guarantee of $\alpha$, for either the secretary or prophet variants of the problem, readily implies the same approximation guarantee for the single-sample secretary problem and single-sample prophet secretary problem respectively \cite{W18}. Indeed, if we consider an instance of the two-sided game of Googol where the values on card $i$ are independent draws from a distribution $F_i$ we exactly obtain the single-sample secretary and the single-sample prophet secretary problems.

\subsection{Our results}

Most of our results come from analyzing three basic algorithms. The first is the {\em Open moving window} algorithm, in which we stop the first time the card just flipped is larger than all the currently visible values. The second is the {\em Closed moving window} algorithm, in which we additionally require that in the last flipped card the value just revealed is larger than the value that was visible before. Finally we consider the {\em Full window} algorithm in which we simply take the largest value initially visible as threshold (and therefore stop the first time we see a value larger than all values we have seen). 

We first study the secretary variant of the two-sided game of Googol in which the goal is to maximize the probability of choosing the maximum hidden value. For this problem we prove that the closed window algorithm gets the maximum value with probability 0.4529. Of course, this value is more than $1/e$ which is the best possible for the classic secretary problem \cite{D63,L61,F89}, but it is less than 0.5801 which is the best possible guarantee for the full information secretary problem, i.e., when the full distribution is known \cite{GM66}. 

Next, we concentrate in the prophet variant of the two-sided game of Googol in which the goal is to maximize the ratio between the expectation of the chosen value, and that of the maximum hidden value. In this case we start by observing that all three algorithms described above can only give a guarantee of 1/2. Indeed, consider an instance with two cards. Card 1 has values 0 and 1 on each side while card 2 has the numbers $\varepsilon$ and $\varepsilon^2$. Clearly the expectation of the maximum hidden value is $1/2+O(\varepsilon)$ (the value 1 is hidden with probability 1/2) while the open moving window algorithm gets in expectation $1/4 + O(\varepsilon)$ (to get the value of 1 the algorithm needs that it is hidden and that card 1 is the first card). For the other two algorithms, consider the instance in which card 1 has the values 1 and $1-\varepsilon$ while card 2 has the values $\varepsilon$ and 0. Clearly the expectation of the maximum hidden value is $1 - O(\varepsilon)$, but since now neither of the algorithms can stop when $1-\varepsilon$ is revealed, both algorithms obtain $1/2+O(\varepsilon)$. However, by randomizing the choice of the algorithm we significantly improve the approximation ratio. 

Our main result in the paper is to show that a simple randomization of the three proposed algorithms achieves a guarantee of $0.635>1-1/e$. Interestingly, our bound surpasses that of Azar et al \cite{ACK18} which until very recently was the best known bound for (the full information) prophet secretary. Furthermore, our bound also beats the bound of $1-1/e$ obtained by Correa et al. \cite{CDFS19} for the i.i.d.~single-sample prophet inequality. So our bound not only improves upon these known bounds, but also works in a more general setting. The key behind the analysis is a very fine description of the performance of each of the three basic algorithms. Indeed for each algorithm we exactly compute the probability that they obtain any of the 2n possible values.\footnote{In this respect, our result for the secretary variant can be seen as a warm up for the prophet variant of the problem.} With this performance distribution at hand it remains to set the right probabilities of choosing each so as to maximize the expectation of the obtained value. 

To wrap-up the paper we consider a large data-set situation which naturally applies to a slightly restricted version of the single sample prophet secretary problem. The assumption states that, with high probability, the if we rank all $2n$ values from largest to smallest, the position of maximum minimum value a card lies far down the list (i.e., when ranking the 2n values the first few values correspond to a maximum in its card). This assumption holds for instance in the single-sample i.i.d.~prophet inequality, under the large market assumption used by Abolhassan et al. \cite{AE17}, or whenever the underlying distributions of the prophet secretary instance overlap enough. For this variant we design an optimal moving window algorithm and prove that it achieves an approximation ratio of $0.642$. 

In table \ref{tab:summary} we summarize our results and the previous best bounds for the problems considered in this paper.

\begin{table}[t]
\centering
\begin{tabular}{|c|c|c|c|}
\hline
\multicolumn{2}{|c|}{Instance}                                                                                                                                 & \begin{tabular}[c]{@{}c@{}}Prophet\\  ($\E(ALG)\geq \alpha \E(OPT)$)\end{tabular}             & \begin{tabular}[c]{@{}c@{}}Secretary\\  ($\PP(ALG=OPT)\geq \alpha$)\end{tabular}                                                             \\ \hline
\multirow{2}{*}{\begin{tabular}[c]{@{}c@{}}Independent\\ (Prophet \\ Secretary)\end{tabular}} & \begin{tabular}[c]{@{}c@{}}Samples\\ $\,$\end{tabular}         & $\frac{1}{2}$ \cite{W18} $\leq 0.635$ [*] $\leq \alpha \leq \ln(2)$ \cite{CDFS19}      & \multirow{3}{*}{\begin{tabular}[c]{@{}c@{}}$\frac{1}{e}$ \cite{D63} $\leq 0.452$ [*] $\leq \alpha$\\ $ \leq 0.580$ \cite{GM66}\end{tabular}} \\ \cline{2-3}
                                                                                              & \begin{tabular}[c]{@{}c@{}}Known\\  distributions\end{tabular} & $0.669$ \cite{CSZ19} $\leq \alpha \leq \sqrt{3}-1$ \cite{CSZ19}                        &                                                                                                                                              \\ \cline{1-3}
\multirow{2}{*}{I.I.D.}                                                                       & \begin{tabular}[c]{@{}c@{}}Samples\\ $\,$\end{tabular}         & $1-\frac{1}{e}$ \cite{CDFS19} $\leq 0.635$ [*] $\leq \alpha \leq \ln(2)$ \cite{CDFS19} &                                                                                                                                              \\ \cline{2-4} 
                                                                                              & \begin{tabular}[c]{@{}c@{}}Known \\ distribution\end{tabular}  & $0.745$\cite{CFHOV17}$\leq \alpha \leq 0.745$ \cite{HK82,K86}                          & $\alpha=0.580$ \cite{GM66}                                                                                                                   \\ \hline
\end{tabular}
\caption{Summary of known results for prophet secretary and secretary problems in the single-sample and full information settings. Results marked [*] are proven in this paper.}
\label{tab:summary}
\end{table}

\subsection{Preliminaries and notation}

\paragraph{Formal problem statement.} We are given $2n$ different and arbitrary positive numbers, organized into $n$ pairs that we denote as $\{(a_i,b_i)\}_{i=1}^n$, representing the numbers written in both sides of each card. These pairs of numbers are shuffled into sets $\UU$ and $\DD$: for each $i=1,...,n$ an independent unbiased coin is tossed, if the coin lands head, then  $a_i=U_i\in \UU$ and $b_i= D_i\in \DD$, otherwise $b_i=U_i \in \UU$ and $a_i=D_i \in \DD$. The set $\UU=\{U_1,...,U_n\}$ represents the numbers that landed facing up, and the numbers in $\DD=\{D_1,...,D_n\}$ are those facing down. The numbers in $\UU$ are revealed, and a random uniform permutation $\sigma\in \Sigma_n$ is drawn. Then the elements in $\DD$ are revealed in $n$ steps: at
each step $i\in[n]$, the value $D_{\sigma(i)}$ and the index
$\sigma(i)$ are revealed. After this, we must decide
whether to stop or to continue to step $i+1$.

We study algorithms (stopping rules) for the two variants. Let $\tau$ be the step at which we stop. In the secretary variant the objective is to maximize $\PP(D_{\sigma(\tau)}= \max \DD)$, and in the prophet variant the objective is to maximize $\E(D_{\sigma(\tau)})/\E(\max \DD)$. Note that the latter is equivalent to just maximizing $\E(D_{\sigma(\tau)})$ since the algorithm does not control $\E(\max \DD)$.

 \paragraph{Ranking and couples} Our analyses rely of the ranking of the 2n numbers of the instance, so let us denote by $Y_1 > Y_2 > \cdots >Y_{2n}$ the $2n$ numbers in $\UU\cup \DD$ ordered from the largest to the smallest. We say that indices $i$ and $j$ are a couple, and denote it by $i\sim j$ if $Y_i$ and $Y_j$ are written in the two sides of the same card, i.e., if $(Y_i,Y_j)=(a_\ell,b_\ell)$ or $(Y_j,Y_i)=(a_\ell,b_\ell)$ for some $\ell\in [n]$. Let $k$ be the smallest index so that $k$ is the couple of some $k'<k$. In particular, $Y_1, Y_2, \dots, Y_{k-1}$ are all written on different cards. Note that $k$ and $k'$ only depend on numbers in the cards\footnote{This notation was introduced by
 Wang~\cite{W18}}, i.e., on the numbers $\{(a_i,b_i)\}_{i=1}^n$, and not on the coin tosses or on the permutation $\sigma$. We will restrict all our analysis to the numbers $Y_1,...,Y_k$, which is easily justified by the observation that $\max \DD$ always lies in the set $\{Y_1,...,Y_k\}$.

\paragraph{Random arrival times.} For most of our analyses it will be useful 
to consider the following reinterpretation of the randomness of both the coins tosses (hidden sides) and the random permutation (flipping order). Consider that each of the $2n$ numbers arrives in a random uniform time in the interval $(-1,1]$ as follows. For every index $\ell$ smaller than its couple $\ell'$ (i.e. $\ell \sim \ell', Y_{\ell} > Y_{\ell'}$) an independent random arrival time $\theta_\ell$ uniform in $(-1,1]$ is sampled and $\ell'$ receives opposite arrival time $\theta_{\ell'}=C(\theta_\ell)$, where 
 \[
C(x)=\begin{cases}x-1,&\text{if $x > 0$,}\\x+1,&\text{otherwise.}\end{cases}
\]
The reinterpretation can the be done as follows. For each $j\in [2n]$, if $\theta_j > 0$, the number $Y_j$ is facing down, i.e., $Y_j\in \DD$, and if $\theta_j\leq 0$ then $Y_j$ is facing up, i.e., $Y_j\in \UU$. Therefore, we get to see all values whose corresponding arrival time is negative and the order in which we scan the hidden values $Y_j \in \DD$ is increasing in $\theta_j$.

Throughout the paper we use the term \emph{value} to refer to the numbers written in the cards $Y_1,Y_2,\ldots,Y_{2n}$, while we use the term \emph{element} to refer to a side of a particular card, that is an element in $\{U_1,\ldots,U_n,D_1,\ldots,D_n\}$. Of course these sets are the same so the point is that the $i$-th value corresponds to $Y_i$, whereas the $i$-th element corresponds to the $i$-th number in the list $U_{\sigma(1)},\ldots,U_{\sigma(n)},D_{\sigma(1)},\ldots,D_{\sigma(n)}$.
\section{Basic Algorithms and statement of our results}
\label{sec:basic_alg}

Once the random coins for all cards and the random uniform permutation are all selected, we are left with the following situation. A list $E=(e_1,\dots, e_{2n})=(U_{\sigma(1)}, \dots, U_{\sigma(n)},D_{\sigma(1)},\dots, D_{\sigma(n)})$ of $2n$ different positive elements are presented one-by-one to an algorithm. The algorithm must observe and skip the first $n$ elements of the list, since they correspond to the values that landed facing up on the cards. The next $n$ elements correspond to values that landed facing down (in such a way that the $s$-th and the $s+n$-th elements form a couple). The algorithm must decide immediately after observing the $s$-th  element whether to select it and stop, or to continue with the next element.  

We now present generic \emph{moving window} algorithms for selecting one of the last $n$ elements of a given list $E=(e_1,\dots, e_{2n})$. These family of algorithms form the basis of the algorithms employed in this paper, which are described next.

\paragraph{Generic moving window algorithm.} 
For every $s\in \{n+1,\dots, 2n\}$, the algorithm first specifies, possibly in a random or implicit way, a \textbf{window} of elements \begin{align}
W_s=\{e_r\colon r\in \{\ell_s, \dots, s-1\}\}
\end{align}  scheduled to arrive immediately before $s$ (the window may be empty). Then the algorithm proceeds as follows. It observes the first $n$ elements without selecting them. For every $s\geq n+1$, the algorithm observes the $s$-th element, $e_s$, and selects it if $e_s$ is larger than every element in its window\footnote{If $W_s=\emptyset$ then we say that $e_s$ is trivially selected.}, i.e. if $e_s> \max W_s$. Otherwise, the algorithm rejects it and continues.

We now present three basic moving window algorithms for deciding when to stop. Each of them can be described by a sequence of left extremes $(\ell_s)_{s=n+1}^{2n}$ that will define the elements inside the window $W_s$.

\paragraph{Open moving window algorithm (\ALGO).} This algorithm corresponds to the following strategy: flip the next unflipped card and accept its value $x$ if and only if $x$ is the largest of the $n$ values that are currently visible (i.e., the $i$-th element is accepted if and only if it is larger than the $n-i$ currently unflipped cards, and is also larger than the $i-1$ previously flipped cards). In terms of the generic moving window, the left extreme of $W_s$ is $\ell_s = s-n+1$ for $n+1\leq s \leq 2n$.

\paragraph{Closed moving window algorithm (\ALGC).} This algorithm works similarly to the previous one, with a slight difference. In this strategy, the element on the back of the last flipped card, the one that was facing up, is also required to be smaller than the element to be selected. In terms of the generic moving window, the left extreme of $W_s$ is $\ell_s = s-n$ for $n+1\leq s \leq 2n$.

\paragraph{Full window algorithm (\ALGF).} This algorithm corresponds to the algorithm presented by Wang~\cite{W18}. It stops when the revealed element is larger than all element seen so far. Note that this is equivalent to stop with the first element that is larger than all elements that landed facing up. In terms of the generic moving window, the left extreme of $W_s$ is $\ell_s = 1$ for $n+1\leq s \leq 2n$.

\paragraph{}Combinig these three algorithms, we can obtain algorithm $ALG^*_r$, which picks \ALGO and \ALGC with probability $(1-r)/2$ and \ALGF with probability $r$. Our main results are:

\begin{restatable}{thm}{SecretaryThm}
\label{thm:main_secretary}
For any instance of the two-sided game of Googol, 
$$\PP(ALG_c = \max \DD)\geq \ln 2 \left(1 - \frac{\ln 2}{2} \right) \approx 0.45292.$$
\end{restatable}

\begin{restatable}{thm}{ProphetThm}
\label{thm:main_prophet}
 There exists a value $r^*\in [0,1]$ such that
 for any instance of the two-sided game of Googol, $\E(\ALG^*_{r^*})\geq
 \alpha \E(\max \DD)$ with $\alpha = \frac{4-5\ln 2}{5-6\ln 2}\approx 0.635184$.
\end{restatable}

The moving windows can also be seen from the perspective of the random arrival times.
For $\ALG_o$, the window for an element that arrives at time $\theta>0$ can be defined as all elements arriving in $(\theta-1,\theta)$. For $\ALG_c$, this changes to $[\theta-1,\theta)$, while for $\ALG_f$ it becomes $[-1,\theta)$. Therefore, we can define
threshold functions\footnote{These thresholds specify the minimum value under which the algorithm stops for each possible arrival time.} for \ALGO, \ALGC, and \ALGF as $T_o(\theta) = \max \{ Y_j: \theta_j \in (\theta-1,\theta)\}$, $T_c(\theta)= \max \{ Y_j: \theta_j \in [\theta-1,\theta)\}$ and $T_f(\theta)= \max \{ Y_j: \theta_j \in [-1,\theta)\}$, respectively.
This motivates the definition of a fourth algorithm, $\ALG(t)$, for which the window
for an element that arrives at time $\theta$ is defined as all elements that arrive in $[\max\{-1,\theta-1-t\}, \theta)$.

Our third result shows that for large $k$, this type of algorithms perform better.

\begin{restatable}{thm}{LargekThm} For $t^*\approx 0.1128904$, and any instance of the two-sided game of Googol, as $k$ goes to infinity, $\E[\ALG(t^*)]/\E[\max \DD]$ is at least $R_\infty(t^*)\approx 0.6426317$.
\end{restatable}

To wrap up this section, we stablish two basic lemmas that will be useful later on. The first is used to estimate the probability of selecting any element arriving before a given element $e_s$. It is clear to see that all our algorithms satisfy the conditions.

\begin{lem}[Moving window lemma]\label{lem:MWLemma}
Suppose the windows' left extremes $(\ell_s)_{s\in \{n+1,\dots 2n\}}$ specified by a moving window algorithm satisfy 
\begin{align}\label{left}
\ell_{n+1}\leq \ell_{n+2}\leq \dots \leq \ell_{2n}\leq n+1
\end{align}
Let $s$ be an index with $n+1\leq s$ and let $e_M=\max \{e_r\colon n+1\leq r\leq s\}$ be the maximum value arriving between positions $n+1$ and $s$.
The algorithm selects some element in the set $\{e_r\colon n+1\leq r\leq s\}$ if and only if $e_M$ is larger than every element in its window, i.e. $e_M>\max W_M$. In particular, if the window $W_s$ of $s$ is nonempty and  $e_J=\max W_s$ then the algorithm does not stop strictly before $e_s$ if and only if $J\leq n$ or ($J\geq n+1$ and $e_J<\max W_J$).
\end{lem}

\begin{proof}
Suppose the algorithm selects some $e_r$ with $n+1\le r\le s$. Note that $e_r$ cannot arrive after $e_M$ because in that case $e_M>e_r$ would be in the window $W_r$ of $e_r$  contradicting that $e_r$ was selected. Then $e_r$ must be $e_M$ or an element arriving before $e_M$. In any case, by condition \eqref{left}, $\ell_r\leq \ell_M$.  Since $e_r$ is selected, $e_r\geq\max\{e_t\colon \ell_r\leq t\leq r-1\}\geq \max\{e_t\colon \ell_M\leq t\leq n\}$. But then, as $e_M$ is the largest element arriving between $n+1$ and $s$, we also conclude that $e_M> \max \{e_t\colon \ell_M\leq t\leq M-1\}=\max W_M$. For the converse suppose that $e_M>\max W_M$ then the algorithm will pick $e_M$ unless it has already selected another element $e_r$ arriving before $e_M$. In any case, the algorithm stops by selecting some element $e_r$, with $n+1\leq r\leq s$.

Let us show the second statement of the lemma. Suppose that the window of $s$ is nonempty and let $J$ be the index of the maximum element in its window. Let also, $M'$ be the largest element arriving between steps $n+1$ and $s-1$. Since $e_{M'}$ is also in the window of $s$, we conclude that either $J\geq n+1$ and $M'=J$, or $J\geq n$ and $e_J>e_{M'}$. Suppose $J\geq n+1$. By the first part of the lemma, not stopping strictly before $e_s$ is equivalent to $e_{M'}<\max W_{M'}$, and since $M'=J$, $e_{J}<\max W_J$. If on the other hand $J\leq n$ then $e_J>e_{M'}$ and by assumption \eqref{left}, $\ell_{M'}\leq \ell_s\leq J\leq n+1\leq M'\leq s-1$. We conclude that $e_J$ is also in the window of $M'$ and therefore $\max W_{M'}\geq e_J>e_{M'}$, which is equivalent to not stopping strictly before $e_s$.\end{proof}

As a final ingredient for our results we need the following lemma about independent uniform random variables in $(-1,1]$.

\begin{lem}\label{lem:independent}
Let $X_0,X_1,\dots, X_m$ be i.i.d.~random variables distributed uniformly in $(-1,1]$. Define the event $E:$ $X_0>0  \text{ and, for all $i\in \{1,\dots, m\}$}, X_i\in [X_0-1,X_0)$.
Conditioned on event $E$,  $\{1-X_0,X_1+1-X_0,X_2+1-X_0,\dots, X_m+1-X_0\}$ is a family of $m+1$ i.i.d.~random variables distributed uniformly in $[0,1)$.
\end{lem}
\begin{proof}
Conditioned on $E$ and on the realization $a\in (0,1]$ of the variable $X_0$, the family $\{X_i+1-a\}_{i\in \{1,\dots, m\}}$ is mutually independent. Since this is true for every realization $a$ of $X_0$, we can uncondition on $a$ and deduce that the family $\{1-X_0,X_1+1-X_0,X_2+1-X_0,\dots, X_m+1-X_0\}$ is mutually independent only conditioned on $E$.  Furthermore, since under $E$, $X_0$ distributes uniformly in $(0,1]$ we conclude that $1-X_0$ is uniform in $[0,1)$. Finally,
under $E$, for $i>0$, $X_i$ distributes uniformly in $[X_0-1,X_0)$, and therefore, $X_i+1-X_0$ distributes uniformly in $[0,1)$.
\end{proof}

\section{Maximizing the probability of picking the maximum}

In this section we present a lower bound for the secretary variant of the two-sided game of Googol. Recall that the objective is to maximize the probability of stopping at the maximum value that landed facing down. If we were maximizing the probability of stopping the overall maximum value (the classical secretary problem), the best strategy \cite{L61, D63} would be to skip a constant number of elements (roughly $1/e$) and then select the first element larger than all previous ones. Since we cannot select any number from the first half, we would be tempted to try the full window algorithm $\ALG_f$ that selects the first element in the second half that is larger than all previous ones. Unfortunately, this algorithm does not select any element if the top one, $Y_1$, appears in the first half. It turns out that the best algorithm among the three basic algorithms presented is the closed moving window algorithm $\ALG_c$ that we analyze below.

\SecretaryThm*

To prove Theorem \ref{thm:main_secretary} the interpretation using random arrival times over $(-1,1]$ will be of use. The idea of the proof is to compute, for a fixed value $k$, the probability of winning, that is, selecting the largest number in $\DD$. To compute this, we partition the event of winning  by identifying which value $Y_i$ is picked and which element $Y_j$ is the largest inside its window. We find out that the worst case is when $k \xrightarrow{} \infty$ and conclude the bound.

Define the events $E_1$: $Y_j$ is the largest element inside $Y_i$'s window, $E_2$: $Y_i=\max\DD$, $E_3$: the algorithm selects $Y_i$. Define $Q_{ij}$ as the event of $E_1$, $E_2$ and $E_3$ happening simultaneously. The following lemma will be of use to compute the probability of winning.

\begin{lem}
\label{lem:Q_ij}

For $1\leq i \leq k-1$,

\[
\PP(Q_{ij})=
\begin{cases}
\frac{1}{2^j} \frac{1}{i} \left(1- \frac{1}{j} \right) &  i+1\leq j \leq k-1,\\
\frac{1}{2^{k-1}} \frac{1}{i} & j=k,\\
0 & \text{otherwise.}
\end{cases}
\]

\end{lem}

\begin{proof}

If $j < i$ and $E_1$ holds,  $Y_i$ will not be selected, therefore $\PP(Q_{ij}|j<i)=0$. Also, note that a window of length 1 will contain either $Y_k$ or $Y_{k'}$, where $k'$ is the couple of $k$, so if $j>k$, $E_1$ can not hold and $\PP(Q_{ij}|j>k)=0$. So assume that $i+1\leq j\leq k$. 

Observe now that $E_1$ is equivalent to the event that  $C(\theta_1), \dots, C(\theta_{i-1})$, $C(\theta_{i+1}),\dots,C(\theta_{j-1}), \theta_j$ are inside the interval $[\theta_i-1,\theta_i)$. As this interval has length $1$ and all variables are independent and uniformly chosen from $[-1,1)$,  $\PP(E_1\,|\,\theta_i>0,j<k)=1/2^{j-1}$. The case $j=k$ is slightly different, because $\theta_k=C(\theta_{k'})$. If $i=k'$ then $\theta_k=\theta_i-1\in [\theta_i-1,\theta_i)$, and if $i\neq k'$ then $\theta_k=C(\theta_l)$ for some $l\in \{1,\dots, j-1\}\setminus \{i\}$. In both cases, we need to impose only $j-2$ variables in $[\theta_i-1,\theta_i)$, resulting in $\PP(E_1\,|\,\theta_i>0,j=k)= 1/2^{k-2}$.

Define the auxiliary events $F_1: \max \{0,C(\theta_1),\dots, C(\theta_{i-1})\}=0$ and $F_2: \max \{0,C(\theta_1),\dots, $ $ C(\theta_{i-1}),C(\theta_{i+1}),\dots, C(\theta_{j-1}), \theta_j)\}\neq \theta_j$. Note that given $E_1$, event  $E_2$ is equivalent to $\theta_1,\dots, \theta_{i-1}\leq 0<\theta_i$, which in turn is equivalent to $\theta_i>0$ and $F_1$. From here $E_1$ and $E_2$ happening simultaneously is equivalent to $E_1, F_1$ and $\theta_i>0$.  Also, by Lemma \ref{lem:MWLemma}, $E_3$ is equivalent to either $\theta_j\leq 0$, or $\theta_j> 0$ and $Y_j$ is smaller than some element arriving in $[\theta_j-1,\theta_j)$.

We claim first that for $j=k$, $Q_{ik}$ is equivalent to $E_1$, $F_1$ and $\theta_i>0$ (which is equivalent to $E_1, E_2$). Indeed, suppose that $E_1, E_2$ hold. If $j=k$ we distinguish 2 cases. If $i\sim k$, then we have that $\theta_k\leq 0$. If $i\not\sim k$ and $\theta_k>0$ then the couple $k'$ of $k$, will satisfy that $Y_{k'}>Y_k$ and $\theta_{k'} \in [\theta_j-1,\theta_j]$. In both cases, $E_3$ holds.

Now we claim that for $1\leq j\leq k-1$, $Q_{ij}$ is equivalent to  all events $E_1, F_1, F_2$ and $\theta_i>0$ happening simultaneously. For that, suppose that $E_1$, and $E_2$ hold. Under that conditioning, $E_3$ holds in two cases, either $\theta_j\leq 0$, or $\theta_j>0$ and for some $\ell\in \{1,\dots, j-1\}\setminus\{i\}$, $\theta_j-1\leq \theta_\ell<\theta_j$, the latter being the same as $C(\theta_\ell)\in [-1,\theta_j-1) \cup [\theta_j,1]$. But recall that by $E_1$, every $\ell \in \{1,\dots, j-1\}\setminus\{i\}$ satisfies $C(\theta_\ell)\in [\theta_i-1,\theta_i)$. Therefore, $E_3$ holds if and only if $\theta_j\leq 0$ or for some $\ell\in \{1,\dots, j-1\}\setminus \{i\}$, $C(\ell)\in [\theta_j,\theta_i)$. This is equivalent to event $F_2$, concluding the claim.

We are ready to compute the probability of $Q_{ij}$. Note that $\theta_j$ being or not the maximum of the set $\{0,C(\theta_1),\dots, C(\theta_{i-1}),C(\theta_{i+1}),\dots, C(\theta_{j-1}),\theta_j\}$ does not depend on the inner ordering of the set $\{0,C(\theta_1),\dots, C(\theta_{i-1})\}=0$. This implies that $F_1$ and $F_2$ are independent events. Since  $F_1$ is equivalent to $\max \{1-\theta_i,C(\theta_1)+1-\theta_i,\dots, C(\theta_{i-1})+1-\theta_i\}=1-\theta_i$, from Lemma \ref{lem:independent}, $\PP(F_1|E_1, \theta_i>0)=1/i$, and since $F_2$ is equivalent to 
$\max\{1-\theta_i,C(\theta_1)+1-\theta_i,\dots, C(\theta_{i-1})+1-\theta_i, C(\theta_{i+1})+1-\theta_i,\dots, C(\theta_{j-1})+1-\theta_i,\theta_j+1-\theta_i\}\neq \theta_j+1-\theta_i$, $\PP(F_2|E_1, \theta_i>0)=1-1/j$. Putting all together and using independence, $\PP(Q_{ik})=\PP(E_1,\theta_i>0,F_1)=\frac12 \frac{1}{2^{k-2}}\frac1i$, and for $1\leq j\leq k-1$
 $\PP(Q_{ij})=\PP(E_1,\theta_i>0,F_1,F_2)=\frac12 \frac{1}{2^{j-1}}\frac1i\left(1-\frac1j\right).$ This concludes the proof of the lemma.\end{proof}

With Lemma \ref{lem:Q_ij}, we can compute the probability of winning as:
\begin{align*}
    \PP(Win) & = \sum_{i=1}^{k-1} \sum_{j=i+1}^{k} \PP(Q_{ij})
     = \sum_{i=1}^{k-1}\left( \sum_{j=i+1}^{k-1} \frac{1}{2^j}\frac{1}{i}\left(1-\frac{1}{j} \right) + \frac{1}{2^{k-1}}\frac{1}{i} \right)
     =\frac{H_{k-1}}{2^{k-1}} + \sum_{i=1}^{k-2} \sum_{j=i+1}^{k-1} \frac{1}{2^j}\frac{1}{i}\frac{j-1}{j},
\end{align*}
where $H_{s} = \sum_{i=1}^s \frac{1}{i}$ is the $s$-th harmonic number. Note that the probability of winning only depends on $k$ and denote it $F(k)$. The following lemma helps to find the worst case scenario for our algorithm.
\begin{lem}
\label{lem:F_decreasing}
$F(k)$ is non-increasing in $k$.
\end{lem}
\begin{proof}
To prove the lemma we note that for $k\geq 2$:
\begin{align*}
F(k+1) - F(k)   & = \frac{H_{k}}{2^{k}}
   - \frac{H_{k-1}}{2^{k-1}}
   +  \sum_{i=1}^{k-1} \frac{1}{2^k}\frac{1}{i}\left(1-\frac{1}{k} \right)\\
                & =\frac{H_{k}}{2^{k}} + \frac{H_{k-1}}{2^k}\left(1 - \frac{1}{k} \right)  - \frac{H_{k-1}}{2^{k-1}} = \frac{1- H_{k-1}}{k2^k} \leq 0.\qedhere
\end{align*}

\end{proof}

\begin{proof}[Proof of \cref{thm:main_secretary}.]
To prove the main theorem of the section we use the previous results to give a lower bound for the probability of winning. The worst case happens $k\xrightarrow{} \infty$:
\begin{align*}
\PP(\text{Win})&\geq \min_{k\geq 2}F_k = \lim_{k\to \infty} F_k=\sum_{i=1}^{\infty}\frac{1}{i}\sum_{j=i+1}^\infty\frac{1}{2^j}\left(1-\frac{1}{j}\right)\\
&=\sum_{i=1}^\infty\frac{1}{i}\sum_{j=i+1}^\infty \left(\frac{1}{2^j} - \int_0^{1/2} x^{j-1}dx\right) =\sum_{i=1}^\infty\frac{1}{i} \left(\frac{1}{2^i}-\int_0^{1/2} \frac{x^i}{1-x}dx\right)\\
&=\sum_{i=1}^\infty \int_0^{1/2} x^{i-1}dx - \int_0^{1/2}\frac{1}{1-x} \int_{0}^x \sum_{i=1}^\infty y^{i-1}dydx\\
&=\int_0^{1/2}\frac{1}{1-x}dx - \int_0^{1/2}\frac{1}{1-x}\int_0^x\frac{1}{1-y}dydx\\
&=\ln 2\left(1-\frac{\ln 2}{2}\right)\approx 0.45292. \qedhere
\end{align*}

\end{proof}
\section{Maximizing the expected value}

In this section we present a lower bound for the prophet variant of the two-sided game of Googol, in which the objective is to maximize $\E(D_{\sigma(\tau)})/\E(\max \DD)$. We analyze a combination of the three basic algorithms presented in \cref{sec:basic_alg}. More precisely, for a fixed $r\in [0,1]$, we define $\ALG^*_r$ as follows. With probability $(1-r)/2$ run $\ALG_o$, with probability $(1-r)/2$ run $\ALG_c$, and with probability $r$ run $\ALG_f$. This combination allows us to obtain the main result of the section.

\ProphetThm*

This ratio is better than the best known guarantee of $1-1/e\approx 0.632$ for the i.i.d.~prophet secretary problem with samples~\cite{CDFS19}, which is
a particular case of our problem.

Recall that none of the 3 basic algorithms can guarantee a ratio better than $\tfrac{1}{2}$ by itself, even in very simple cases. Indeed, as shown in the introduction, for the instance $\{(1,0),(\varepsilon,\varepsilon^2)\}$ we have that $\E(\max \DD)=1/2+O(\varepsilon)$ and $\E(\ALGO)=1/4 + O(\varepsilon)$. Furthermore, if we consider the instance $\{(1,1-\varepsilon),(\varepsilon,0)\}$ then $\E(\max \DD)=1 - O(\varepsilon)$ but $\E(\ALGC)=\E(\ALGF)=1/2+O(\varepsilon)$.
 
The intuition of our result is that the situations where each of the three algorithms perform poorly are very different. As we are able to compute exactly the distribution of the outcome of each algorithm, we can balance their distributions, so that in all cases $ALG^*_r$ performs well. We can state this in terms of the
ordered values $Y_1,...,Y_k$. On the one hand, since \ALGF is very restrictive, it has good chances of stopping at $Y_1$ when it is in $\DD$. On the other hand, \ALGC has higher probability of stopping at $Y_i$ than \ALGF, for $1<i<k$. But none of the two algorithms can stop at $Y_k$, whereas \ALGO stops at
$Y_k$ with positive probability. This comes at the cost of sacrificing a bit of the better elements, but as showed in the examples, it is very important for the case when $k$ is small.

Roughly speaking, the proof of \cref{thm:main_prophet} goes as follows. For each of the basic algorithms we compute in \cref{lem:dist_algc,lem:dist_algo,lem:dist_algf}  the probability that they stop with the value $Y_j$ for each $j<k$, and express it as a series truncated in $k$ and an extra term that depends on $k$. Then in
\cref{lem:tail_series} we prove that in the combined formula, for certain values of $r$, the extra term can be replaced with the tail of the series, so the formula does not depend on $k$. Finally, for a specific value of $r$, we show in \cref{lem:geometric_distribution,lem:k_geometric_distribution} an approximate stochastic dominance and conclude, i.e., we make use of the following general observation.
\begin{obs}
If $\min_{j\leq k}\frac{\PP(\ALG \geq Y_j)}{\PP(\max \DD\geq Y_j)}\geq \alpha$, then $\E(\ALG)\geq \alpha \E(\max \DD)$.\footnote{Recall that $\max\DD$ can only take values in $\{Y_1,...,Y_k\}$, so the implication follows immediately.}
\end{obs}

\subsection{Distribution of the Maximum}

As noted before, $\PP(\max \DD = Y_i) = 0$ if $i>k$. For the remaining values, Wang \cite{W18} did a simple analysis for the distribution of $\max \DD$.

For $Y_i$ to be the maximum of $\DD$ we need: (1) $Y_i\in \DD$ and (2) $Y_j \in \UU$ for $1\leq j < i$. If $1\leq i < k$ all these events are independent with probability $\tfrac{1}{2}$, so we have $\PP(\max \DD =Y_i) = (\tfrac{1}{2})^i$. If $i=k$, the fact that $Y_k$ is facing downwards implies that its couple is facing upwards, so we need one less coin toss for the events to happen simultaneously. Putting all together:

\[\mathbb{P}(\max \DD=Y_i)= \begin{cases} 
      \frac{1}{2^i} & i<k \\
      \frac{1}{2^{k-1}} & i=k\\
      0 & i> k 
   \end{cases}
\]

Therefore we have that, $\PP(\max \DD\geq Y_i) = 1 - \tfrac{1}{2^i}$ if $1 \leq i < k$  and $\PP(\max \DD\geq Y_k) = 1$.


\subsection{Analysis of the basic algorithms}
In this section we precisely derive the distribution of the obtained value of each of the basic algorithms. In the next section we combine these distributions to find an improved randomized algorithm. In what follows, we will denote by $k'$ the couple of $k$. Note that the identities of $k$ and $k'$  depend only on the instance, and not on the realizations of the random coins or the random permutation.

\begin{lem}
\label{lem:dist_algc}
For every $1\leq i \leq k-1$,
\begin{align*}
    \PP(\ALGC=Y_i)&=\frac{1}{2^{k-1}} + \sum_{j=i+1}^{k-1}\frac{1}{2^j}\left(1-\frac{1}{j}\right)\\
\shortintertext{and }\qquad \PP(\ALGC\geq Y_k)&= \sum_{j=1}^{k-1}\frac{1}{2^j j}.
\end{align*}
\end{lem}
\begin{proof}
Intuitively, for each $j>i$, we condition on that $Y_j$ is the maximum value in the window of $Y_i$ and show that \ALGC stops in $Y_i$ with probability $(1-1/j)$. More precisely, fix an index $i\leq k-1$. For $j\leq i, j\neq i$, denote by $F_j$ the event that $\theta_i\geq 0$ (or equivalently $Y_i\in \DD$) 
and that $j$ is the smallest index such that $\theta_j\in [\theta_i-1,\theta_i)$.
If $j<k$ then $\PP(F_j)=1/2^j$ because the variables $\{\theta_{\ell}\}_{\ell\leq j}$
are all independent and uniformly distributed in $(-1,1]$, and for any given $\theta_i\geq 0$, the interval $[\theta_i-1,\theta_i)$ has length $1$. For $F_k$ we have that
$\theta_k\in [\theta_i-1,\theta_i)$ if and only if $\theta_{k'}= C(\theta_k) \not\in [\theta_i-1,\theta_i)$. Since the variables $\{\theta_{\ell}\}_{\ell\leq k, \ell\neq k'}$
are all independent, $\PP(F_k)= 1/2^{k-1}$.

Note that if for some $j<i$, $\theta_j\in [\theta_i-1,\theta_i)$, then $\ALGC\neq Y_i$,
because its threshold $T_c(\theta_i)$ will be at least $Y_j>Y_i$. Also note that either $\theta_k$ or $\theta_{k'}$ is in $[\theta_i-1,\theta_i)$, so the events $(F_j)_{j=i+1}^k$ (that are pairwise disjoint) completely
cover the event $\{\ALGC=Y_i\}$. Thus we have the identity
\begin{align}
    \PP(\ALGC=Y_i) = \frac{1}{2^{k-1}}\PP(\ALGC=Y_i| F_k) + \sum_{j=i+1}^{k-1} 
    \frac{1}{2^{j}}\cdot  \PP(\ALGC=Y_i| F_j).
\end{align}
Now, for $j>i$, conditional on $F_j$, $\ALGC= Y_i$ if and only if $\ALGC$ does not stop
before observing $Y_i$. Using \cref{lem:MWLemma,lem:independent} we will show that  $\PP(\ALGC=Y_i|F_k)=1$ and that $\PP(\ALGC=Y_i|F_j)= (1-1/j)$ if
$i<j<k$, and conclude the first formula of the lemma.

Conditional on $F_k$, the maximum value in the window of $Y_i$ is $Y_k$, but
$Y_k$ is not larger than all elements in its window, because $Y_{k'}$ is in it (it is on the other side of its card). So \cref{lem:MWLemma} implies that $\PP(\ALGC=Y_i|F_k)=1$.

For the case $i<j<k$, we have that conditional on $F_j$, $Y_j$ is the maximum value in the window of $Y_i$, so \cref{lem:MWLemma} implies that \ALGC stops in
$(0,\theta_i)$ if and only if $\theta_j\geq 0$ and $Y_j$ is larger than all elements in its window, i.e., $\theta_{\ell}\not \in [\theta_j-1,\theta_j)$, for all $\ell<j$. This is equivalent to $\max \{0, C(\theta_1), ..., C(\theta_{i-1}),
C(\theta_{i+1}),..., C(\theta_{j-1}), \theta_j\} = \theta_j$. Now, \cref{lem:independent} implies that
conditional on $F_j$, the random variables
$\{0-(\theta_i-1), C(\theta_1) -(\theta_i-1) , ..., C(\theta_{i-1})-(\theta_i-1),
C(\theta_{i+1})-(\theta_i-1),..., C(\theta_{j-1})-(\theta_i-1), \theta_j-(\theta_i-1)\}$ are independent and uniformly distributed in $(0,1)$. Hence,
conditional on $F_j$, \ALGC stops in $(0,\theta_i)$ with probability $1/j$, so
$\PP(\ALGC=Y_i|F_j)=(1-1/j)$. This proves the first formula.

For computing $\PP(\ALGC\geq Y_k)$, which is the probability that \ALGC stops, define as $E_j$ the event that $j$ is the smallest index such that $\theta_j>0$ (or equivalently, that $Y_j=\max \DD$). It is clear that for $j<k$, $\PP(E_j)= 1/2^j$, that $\PP(E_k)=1/2^{k-1}$, and that $\PP(E_j)=0$ if $j>k$.
For $j<k$ , conditional on $E_j$, \cref{lem:MWLemma} implies that \ALGC stops if and only if $\max \{C(\theta_1), C(\theta_2),...,C(\theta_{j-1}),\theta_j\} =\theta_j$. But this happens with probability $1/j$. Conditional on $E_k$, \ALGC never stops. This concludes the proof of the lemma.
\end{proof}

\begin{lem}
\label{lem:dist_algo}
For every $1\leq i \leq k-1$,
\begin{align*}\PP(\ALGO=Y_i)&=\frac{1}{2^{k-1}}\left(1-\frac{\mathds{1}_{\{i\neq k'\}}}{k-1}\right) + \sum_{j=i+1}^{k-1}\frac{1}{2^j}\left(1-\frac{1}{j}\right)
\shortintertext{and}
\PP(\ALGO\geq Y_k) &= \frac{1}{2^{k-1}(k-1)} + \sum_{j=1}^{k-1} \frac{1}{2^j j}.
\end{align*}
\end{lem}

\begin{proof}
Recall that \ALGO cannot select a value $Y_j$ with $j>k$. Suppose now that we run \ALGO and \ALGC in the same instance and realization. Since \ALGC is more restrictive than \ALGO, if \ALGO stops, it stops earlier than \ALGC. Nevertheless, note that $Y_k$ is the only value that can be accepted by \ALGO but not by \ALGC. Hence, for $1\leq i\leq k-1$, if $\ALGO = Y_i$ then $\ALGC=Y_i$, and if $\ALGC= Y_i$ then either $\ALGO=Y_i$ or $\ALGO=Y_k$. Thus, we can write the following identity for the case $1\leq i \leq k-1$.
\begin{align}
    \PP(\ALGO=Y_i) = \PP(\ALGC =Y_i) - \PP(\ALGC=Y_i \text{ and } \ALGO =Y_k)\, . \label{eq:ALGO_and_ALGC}
\end{align}
Thus, we just need to compute the negative term in \cref{eq:ALGO_and_ALGC}.
In order to have $\ALGO=Y_k$ we need that $\theta_k>0$ and that
$\theta_j \not\in (\theta_k-1,\theta_k]$, for all $1\leq j \leq k-1$, with $j\neq k'$. Note that this is also a sufficient condition, because values that are smaller than $Y_k$ cannot be accepted by \ALGO. Thus, $\PP(\ALGO=Y_k)= 1/2^{k-1}$. If $i=k'$, then $\theta_i=\theta_{k'}=\theta_k-1$, which is negative if $\theta_k$ is positive,
so $\PP(\ALGC= Y_{k'}$ and $\ALGO=Y_k)=0$.

Assume now that $i\neq k'$. Conditional on $\ALGO=Y_k$, we have that $\ALGC=Y_i$ if $Y_i$ arrives after $Y_k$ and all other values in $\{Y_1,...,Y_k\}$ arrive either before the window of $Y_k$ or after $Y_i$, i.e., if
$\theta_k<\theta_i$ and $\theta_j \in (-1,\theta_k-1)\cup (\theta_i,1]$, for all $j\leq k-1$, with $j\not\in \{i,k'\}$. This is equivalent to say that $C(\theta_i) =
\min \{0, C(\theta_1),...,C(\theta_{k'-1}),C(\theta_{k'+1}),..., C(\theta_{k-1})\}$. Note that \cref{lem:independent} implies that conditional on $\ALGO=Y_k$, the variables
$\{0-(\theta_k-1), C(\theta_1)-(\theta_k-1),...,C(\theta_{k'-1})-(\theta_k-1),C(\theta_{k'+1})-(\theta_k-1),..., C(\theta_{k-1})-(\theta_k-1)\}$ are independent and uniformly distributed in $[0,1]$. Then,
$\PP(\ALGC=Y_i|\ALGO =Y_k)=1/(k-1)$.

We conclude the first formula by replacing the just computed probability and the formula of \cref{lem:dist_algc} in \cref{eq:ALGO_and_ALGC}. For the second one, we have that $\PP(\ALGO\geq Y_k)= \sum_{i=1}^k
\PP(\ALGO=Y_i) = \PP(\ALGO=Y_k)- (k-2)\frac{1}{2^{k-1}(k-1)}+  \sum_{i=1}^{k-1}\PP(\ALGC=Y_i)=\frac{1}{2^k}- (k-2)\frac{1}{2^{k-1}(k-1)} +\PP(\ALGC\geq Y_k)
 = \PP(\ALGC\geq Y_k) + \frac{1}{2^{k-1}(k-1)}$.
\end{proof}

\begin{lem}
\label{lem:dist_algf}
For every $1\leq i \leq k-1$,
\begin{align*}
\PP(\ALGF=Y_i)&=\frac{1}{2^{k-1}}\frac{1}{k-1} + \sum_{j=i+1}^{k-1}\frac{1}{2^j(j-1)}
\shortintertext{and}
\PP(\ALGF\geq Y_k) &= \frac12\,.
\end{align*}
\end{lem}

\begin{proof}
Let $F_j$ be the event that $Y_j$ is the maximum value ($j$ is the minimum index) such that $\theta_j<0$. Roughly speaking, we condition on each $F_j$ and show that $\ALGF=Y_i$ if $Y_i$ is the value with earliest arrival time from $\{Y_1,...,Y_{j-1}\}$.

Note that the events $F_j$ are pairwise disjoint, and since either $\theta_k$ or $\theta_{k'}$ is negative, $\PP(F_j)=0$ if $j>k$. Thus, the events $F_j$ for $1\leq j\leq k$ form a partition of the probability space. Note also that conditional on $F_j$, by definition \ALGF cannot accept any value smaller than $Y_j$. Therefore
$\PP(\ALGF=Y_i)= \sum_{j=i+1}^k \PP(\ALGF=Y_i| F_j) \PP(F_j)$.

For $j\leq k-1$ the arrival times $\theta_1,...,\theta_j$ are all independent and uniform in $(-1,1]$, so $\PP(F_j)=1/2^j$. Conditional on $F_j$ we have that $\theta_1,...,\theta_{j-1}$ are independent and uniform in $[0,1]$. Moreover, conditional on $F_j$, \ALGF simply accepts the element in $\{Y_1,...,Y_{j-1}\}$ with earliest arrival time, as they are exactly the ones larger than $\max \UU= Y_j$. Hence, $\PP(\ALGF =Y_i| F_j)= 1/(j-1)$.

The case of $F_k$ is slightly different. $F_k$ is equivalent to the event that $\theta_1,...,\theta_{k-1}>0$, because $\theta_{k'}>0$ implies that $\theta_k<0$, so $\PP(F_k)=1/2^{k-1}$. Again \ALGF simply accepts the first element larger than $Y_k$, so it accepts $Y_i$ with probability $1/(k-1)$. This concludes the proof of the first formula.

For computing $\PP(\ALGF\geq Y_k)$, note that this is simply the probability that $\theta_1>0$, which is $1/2$. This is because $Y_1$ is the overall largest value, so if it arrives before $0$, nothing can be accepted, and if it arrives after $0$, then $Y_1$ itself (and possibly other values with earlier arrival times) can be accepted. 
\end{proof}

\subsection{The combined algorithm}
We now use the distributions obtained in the last section in order to obtain an improved randomized algorithm. To this end we combine \cref{lem:dist_algo,lem:dist_algc,lem:dist_algf} to conclude that for $1\leq i \leq k-1$,
\begin{align}
    \PP( \ALG^*_r=Y_i)&= \frac{1-r}{2^{k-1}} +\frac{r-\mathds{1}_{\{i\neq k'\}}\cdot (1-r)/2}{2^{k-1}(k-1)} + \sum_{j=i+1}^{k-1}\frac{1}{2^j}\left((1-r) - \frac{(1-r)}{j} + \frac{r}{j-1}\right) \notag \\
    &\geq 
    \frac{1-r}{2^{k-1}} +\frac{r- (1-r)/2}{2^{k-1}(k-1)} + \sum_{j=i+1}^{k-1}\frac{1}{2^j}\left((1-r) - \frac{(1-r)}{j} + \frac{r}{j-1}\right)\,.
    \label{eq:dist_algr}
\end{align}
The inequality comes from the fact that $-\mathds{1}_{\{i\neq k'\}}\geq -1$. 
Now we use the following lemma to complete the summation in \cref{eq:dist_algr} using the extra term.

\begin{lem}
\label{lem:tail_series}
For every $\frac{3-4\ln(2)}{5-6\ln(2)}\leq r\leq 2/3$, and any $k\geq 2$,
\[\sum_{j\geq k} \frac{1}{2^j}\left(1-r - \frac{1-r}{j} + \frac{r}{j-1}\right) \leq \frac{1-r}{2^{k-1}}+\frac{r-(1-r)/2}{2^{k-1}(k-1)}\, .\]
\end{lem}

\begin{proof}
Rearranging terms, we obtain the following.
\begin{align}
\sum_{j\geq k} \frac{1}{2^j}\left(1-r - \frac{1-r}{j} + \frac{r}{j-1}\right) &=\frac{1-r}{2^{k-1}}-\sum_{j\geq k}\frac{1-r}{2^j j}+\sum_{j\geq k-1}\frac{r/2}{2^j j}\notag \\
&=\frac{1-r}{2^{k-1}}+\frac{r/2}{2^{k-1}(k-1)}-\sum_{j\geq k}\frac{1-3r/2}{2^j j}\notag \\
&=\frac{1-r}{2^{k-1}}+\frac{r/2}{2^{k-1}(k-1)}-\frac{(1-3r/2)}{2^{k-1}(k-1)}\sum_{j \geq 1 }\frac{(k-1)}{2^{j}(j+k-1)}.\label{EQ1}
\end{align}
Therefore, it is enough to prove that $-(1-3r/2)\sum_{j \geq 1 }\frac{(k-1)}{2^{j}(j+k-1)}
\leq r-1/2$.
Since $r\leq 2/3$, the term $(1-3r/2)$ is positive. Note that for $k\geq 2$, we have $\sum_{j \geq 1 }\frac{(k-1)/(j+k-1)}{2^{j}}\geq \sum_{j \geq 1 }\frac{1}{2^{j}(j+1)}=2\ln(2)-1$.
Thus, it would be enough to prove that $-(1-3r/2)(2\ln 2-1) \leq r-1/2$. Rearranging the terms in
the last expression and multiplying by $-1$, we obtain the equivalent condition $3-4\ln 2 \leq r(5 - 6\ln 2)$.
\end{proof}

In what follows, we apply \cref{lem:tail_series} to derive a bound on the distribution of the accepted element that does not depend on $k$. Then, we select a specific value for $r$ that balances these bounds, and allows us to approximate the distribution of $\max \DD$. We define now the function
\begin{align}
    a(r):= \sum_{j\geq 2} \frac{1}{2^j} \left( (1-r)-\frac{(1-r)}{j} + \frac{r}{j-1}
    \right)
    = 1- \ln 2 + \frac{r(3\ln 2 - 2)}{2}.
    \label{eq:def_a}
\end{align}
This function appears in the next lemma we prove as the approximation factor of the distribution of $\max\DD$.
\begin{lem}
\label{lem:geometric_distribution}
For $r^*= \frac{3-4\ln 2}{5-6\ln 2}$ we have that for all $1\leq i \leq k-1$,
\begin{align}
    \PP( ALG^*_{r^*} = Y_i) \geq \frac{1}{2^{i}} 2 a(r^*)\; .
    \label{eq:lemma_geometric_distribution}
\end{align}
\end{lem}

\begin{proof}
Using \cref{lem:tail_series} in \cref{eq:dist_algr} we can write the following for $r^* \leq r\leq 2/3$.
\begin{align}
    \PP(ALG^*_{r} =Y_i) \geq \sum_{j\geq i+1}
    \frac{1}{2^j} \left( (1-r)-\frac{(1-r)}{j}+ \frac{r}{j-1} \right)\,.
    \label{eq:simpler_dist_algr}
\end{align}
Note that this immediately gives the desired bound for $i=1$ and any $r$ in the interval $[r^*,2/3]$. We will prove directly that it also holds for $i=2$ and will proceed by induction for $i\geq 3$. For $i=2$, \cref{eq:simpler_dist_algr} combined with the definition of $a(r)$ gives that
\begin{align*}
    \PP(ALG^*_{r} =Y_2)  & \geq a(r) - \frac{1}{4}\left( (1-r)
    -\frac{(1-r)}{2} + r \right)\\ 
    & = a(r) - \frac{1+r}{8}\,.
\end{align*}
Therefore, it is enough to prove that $a(r^*) -\frac{1+r^*}{8}\geq a(r^*)/2$, which is equivalent to $4 a(r^*)\geq 1+r^*$. If we replace $a(r^*)$ with the explicit formula in the right hand of \cref{eq:def_a} and rearrange terms, we obtain the inequality
$3-4\ln(2) \geq r^*(5-6\ln(2))$. Note that this is satisfied with equality by $r^*$.

For $i\geq 3$ we simply prove that if $r\leq 1/3$, then
\begin{align}
    \sum_{j\geq i+1}
    \frac{1}{2^j} \left( (1-r)-\frac{(1-r)}{j}+ \frac{r}{j-1} \right)
    \geq \frac{1}{2} \sum_{j\geq i}
    \frac{1}{2^j} \left( (1-r)-\frac{(1-r)}{j}+ \frac{r}{j-1} \right)\,. \label{eq:recurrence_sum}
\end{align}
In fact, note that we can change the index in the right-hand side of \cref{eq:recurrence_sum} to get the same range as in the sum of the left-hand side. So when we write the inequality for each term of the two summations, we obtain
\begin{align*}
    \frac{1}{2^j}\left( (1-r) -\frac{1-r}{j} + \frac{r}{j-1} \right) &\geq 
    \frac{1}{2^j}\left( (1-r) - \frac{1-r}{j-1} + \frac{r}{j-2} \right)& \Leftrightarrow \\
    \frac{1}{j-1} - \frac{1}{j} & \geq r \left( \frac{1}{j-2} - \frac{1}{j} \right) & \Leftrightarrow \\
    \frac{j-2}{2(j-1)} &\geq r\,. &
\end{align*}
which holds whenever $j\geq 4$ and $r\leq 1/3$. Since $r^*=\frac{3-4\ln 2}{5-6\ln 2} \approx 0.270$, we conclude that
\cref{eq:recurrence_sum} holds for $r^*$, and therefore, \cref{eq:lemma_geometric_distribution} holds for all $1\leq i \leq k-1$.
\end{proof}

\begin{lem}
\label{lem:k_geometric_distribution}
For any $k\geq 2$, we have that $\PP( ALG^*_{r^*} \geq Y_k) \geq  2 a(r^*)$.
\end{lem}

\begin{proof}
 Using \cref{lem:dist_algo,lem:dist_algc,lem:dist_algf} we get the following inequality.
 \begin{align}
     \PP( ALG^*_{r^*} \geq Y_k) = (1-r^*)\left(
     \frac{1}{2^k(k-1)} + \sum_{j=1}^{k-1} \frac{1}{2^j j}
     \right) + \frac{r^*}{2}\,. \label{eq:stop_prob_ineq}
 \end{align}
 Define now the function $G(k)=  \frac{1}{2^k(k-1)} + \sum_{j=1}^{k-1} \frac{1}{2^j j}$. We minimize $G$ to obtain a general lower bound for \cref{eq:stop_prob_ineq}. We have that
 \begin{align*}
     G(k+1)-G(k) &= \frac{1}{2^k k} + \frac{1}{2^{k+1} k} - \frac{1}{2^k (k-1)} \\
     &= \frac{2(k-1)+ (k-1) - 2k}{2^{k+1}k(k+1)}\\
     &= \frac{k-3}{2^{k+1}k(k+1)}\,.
 \end{align*}
 Hence, $G(2)\geq G(3)$, and $G(k)\leq G(k+1)$ for $k\geq 3$, so $G$ is minimized when $k=3$. Thus,
 \begin{align*}
    \PP( ALG^*_{r^*} \geq Y_k) &\geq (1-r^*)\left(
    \frac{1}{16} + \frac{1}{2} + \frac{1}{8} \right)
    + \frac{r^*}{2}\\
    &= \frac{11-3r^*}{16}\\
    &= \frac{55-66\ln 2 -9 + 12\ln 2}{16(5-6\ln 2)}\\
    &= \frac{16(4-5\ln 2)}{16(5-6\ln 2)} +
    \frac{26\ln 2 - 18}{16(5-6\ln 2)}\\
    &\geq 2 a(r^*).
 \end{align*}
 The last inequality comes from the fact that
 $\frac{26\ln 2 -18}{16(5-6\ln 2)}\approx 0.00162\geq 0$ and that
 \begin{align*}
 a(r^*)&=1-\ln 2 + \frac{r(3\ln 2 -2)}{2}\\
 &= \frac{(2-2\ln 2)(5-6\ln 2) + (3-4\ln 2)(3\ln 2-2)}{2(5-6\ln 2)}\\
 &= \frac{4-5\ln 2}{2(5-6\ln 2)}.\qedhere
 \end{align*}
\end{proof}
With the last two lemmas we are ready to prove the main theorem of this section.
\begin{proof}[Proof of \cref{thm:main_prophet}.]
Summing the lower bound in \cref{lem:geometric_distribution} it follows that
for all $1\leq i \leq k-1$,
\begin{align*}
    \PP(ALG^*_{r^*}\geq Y_i) & \geq 2 a(r^*) \cdot \left(
    1- \frac{1}{2^i} \right)\\
    &= 2 a(r^*) \PP(\max\DD \geq Y_i)\,.
\end{align*}
From \cref{lem:k_geometric_distribution} we get that also $\PP(ALG^*_{r^*}\geq Y_k)\geq 2a(r^*)\cdot
\PP(\max\DD\geq Y_k)$. Therefore, as $\max\DD\in \{Y_1,...,Y_k\}$ with probability $1$, we conclude that $\E(ALG^*_{r^*})\geq 2a(r^*)\cdot \E(\max\DD)$.
\end{proof}

\section{Large Data Sets (large $k$)}

In this section we consider the case in which $k$ is large and show that for that case one can obtain better guarantees for the prophet two-sided game of Googol. This case appears very often, for instance, in the i.i.d.~prophet secretary problem with unknown (continuous) distributions, all permutations of $[2n]$ are equally likely to be the ordering of the $2n$ values in the cards. The probability that $k$ is at least $k_0$ equals the probability that the top $t$ elements of the list are written in different cards. Note that for any $\ell=o(\sqrt{n})$, 
\begin{align*}
\PP(k\geq \ell)&=\left(1- \frac{1}{2n-1}\right)\cdot \left(1- \frac{2}{2n-2}\right)\cdots \left(1- \frac{\ell-1}{2n-(\ell-1)}\right)\geq \left(1-\frac{\ell}{2n-\ell}\right)^{\ell}\\
&\geq 1-\frac{\ell^2}{2n-\ell}\geq 1-o(1).
\end{align*}

The algorithm that we use uses a moving window of length strictly larger than 1, as outlined in Section \ref{sec:basic_alg}.

\paragraph{Moving window algorithm of length $1+t$ (\ALG(t)).}  We first draw $n$ uniform variables in $[0,1]$ and we sort them from smallest to largest\footnote{We assume that the drawn values are all distinct as this happens with probability 1}  as $0<\tau_1<\tau_2<\dots<\tau_n<1$. We interpret $\tau_j$ as the \emph{arriving time} of the $s$-th hidden value that we reveal, and therefore, $\tau_s-1 = C(\tau_j)$ is the arriving time of the corresponding $s$-th face up value. The algorithm will accept the $s$-th face up value $x=D_{\sigma(s)}$ if and only if $x$ is larger than any value arrived in the previous $1+t$ time units, i.e. if and only if $x$ is larger than all elements arriving in $[\max(\tau_s-1-t,-1), \tau_s]$

Observe that $\ALG(0)$ and $\ALG(1)$ are exactly the algorithms $\ALGC$ and $\ALGF$ defined in previous sections. Below we analyze the performance of $\ALG(t)$.

\begin{lem}
\label{lem:dist_algt}
For every $1\leq i \leq k-1$,
\begin{align}\PP(\ALG(t)=Y_i)&=b(k,t) +\sum_{j=i+1}^{k-1}\frac{a(j,t)}{2^j} \label{dist1},\\
\text{ where }\quad  a(j,t)&=  \frac{1-(1-t)^{j-1}(1+t)}{j-1} + (1-t)^{j-1}(1+t) -\frac{(1-t)^{j}}{j},\nonumber\\
b(k,t)&=\frac{1}{2^{k-1}}\left( \frac{1-(1-t)^{k-1}}{k-1} + (1-t)^{k-1}\right),
\notag\shortintertext{and}
\PP(\ALG(t)\geq Y_k)&=\frac{1}{2} + \sum_{i=2}^{k-1}\frac{(1-t)^{j}}{2^{j} j}.\nonumber
\end{align}
\end{lem}
\begin{proof}
Fix an index $i\leq k-1$. Recall that $\theta_i\in [-1,1)$ denotes the arriving time of $Y_i$. Let $J$ be the random variable denoting the index of the largest element arriving in the window of $Y_i$, or equivalently, the smallest index such that  $\theta_j\in [\theta_i-1-t,\theta_i)$. Since $[\max(\theta_i-1-t,-1),\theta_i)$ has length at least $1$, it always contains $\theta_k$ or $C(\theta_k)=\theta_{k'}$ (or both), we conclude that $J\leq k$.

Note that the algorithm selects $Y_i$ if and only if the next three events occur: (I) $\theta_i> 0$, (II) $J>i$, as otherwise, $Y_J>Y_i$  would be inside the window of $Y_i$, and (III) No  element is selected before $Y_i$'s arrival. By Lemma \ref{lem:MWLemma}, the third event occurs if and only if the maximum element\footnote{If no element arrives in $[0,\theta_i)$ then event (III) also occurs} $X$ arriving in $[0,\theta_i)$ is smaller than some element in its own window. Note that this happens in two cases: either $\theta_J<0$ (and therefore $Y_J$ is in the window of $X$), or $\theta_J\geq 0$ and $Y_J$ is smaller than some element arriving in its own  window.

In what follows fix some $j$ with $i+1\leq j\leq k-1$. We will compute the probability of both selecting $Y_i$ and $J=j$.

\paragraph{Case 1: $0\leq \theta_i\leq t$.} In order to impose that $J=j$ we need that all $j-2$ arrival times $(\theta_l: l\in [j-1]\setminus\{i\})$ are outside the interval $[-1,\theta_i]$ and $\theta_j$ is inside. Furthermore, in this case every element arriving in $[0,\theta_j)$ is strictly smaller than $Y_j$. Therefore in order to impose that no element is selected before $Y_i$'s arrival we need that $\theta_j<0$. It follows that
\begin{align*}
\PP(\ALG(t)=Y_i, J=j, \theta_i\in [0,t])&=\frac12\int_0^t\PP(\theta_j<0)\prod_{l\in [j-1]\setminus \{i\}}\PP(\theta_l>1-x)dx\\
&=\frac12\int_0^t\frac12 \left(\frac{1-x}{2}\right)^{j-2}dx = \frac{1}{2^j}\cdot\frac{1-(1-t)^{j-1}}{j-1}.
\end{align*}

\paragraph{Case 2: $t< \theta_i\leq 1$ and $\theta_j<0$.} We now need that all $j-2$ arrival times $\theta_l, l\in [j-1]\setminus\{i\}$ fall outside $[\theta_i-t-1,\theta_i]$ and $\theta_j$ is inside. Since we are imposing $\theta_j<0$, we actually require that $\theta_j\in [\theta_i-t-1,0]$. Since $\theta_J<0$ this is enough to guarantee that $Y_i$ will be selected. Therefore,
\begin{align*}
&\PP(\ALG(t)=Y_i, J=j, \theta_i\in (t,1] \wedge \theta_j<0)\\
&=\frac12\int_t^1\PP(\theta_j\in [x-t-1,0])\prod_{l\in [j-1]\setminus \{l\}}\PP(\theta_l\not\in [x-t-1,x])dx\\
&=\frac12 \int_t^1\frac{1+t-x}{2} \left(\frac{1-t}{2}\right)^{j-2}dx =   \frac{(1-t)^{j-1}(1+t)}{2^{j+1}}.
\end{align*}

\paragraph{Case 3: $t< \theta_i\leq 1$ and $0\leq \theta_j<t$.} As before we need that all $j-2$ arrival times $\theta_l, l\in [j-1]\setminus\{i\}$ fall outside $[\theta_i-t-1,\theta_i]$ and $\theta_j\in [0,t)$. Since $\theta_j\geq 0$ we also need that at least one of $Y_l, l\in [j-1]\setminus \{j\}$ arrives in $Y_j$'s window $[-1,\theta_j)$ and since they must be outside $Y_i$'s window this reduces to the event that not all $j-2$ arrival times fall in the interval $(\theta_i,1]$. Therefore,
\begin{align*}
&\PP(\ALG(t)=Y_i, J=j,  \theta_i\in (t,1] \wedge 0\leq \theta_j<t)\\
&=\frac12\int_t^1\PP(\theta_j\in [0,t))\PP(\forall l\in [j-1]\setminus \{i\}, \theta_l\not \in [x-t-1,x] \text{ and not all in } (x,1]) dx\\
&=\frac12\int_t^1\frac{t}{2} \left(\left(\frac{1-t}{2}\right)^{j-2}-\left(\frac{1-x}{2}\right)^{j-2}\right)dx=\frac{t(1-t)^{j-1}}{2^j} - \frac{t(1-t)^{j-1}}{2^j(j-1)}
\end{align*}

\paragraph{Case 4: $t< \theta_i\leq 1$ and $t\leq \theta_j<\theta_i$.} Once again we need that all $j-2$ arrival times $\theta_l, l\in [j-1]\setminus\{i\}$ fall outside $[\theta_i-t-1,\theta_i]$ and $\theta_j\in [t,\theta_i)$. Since $\theta_j\geq 0$ we also need that at least one of $Y_l, l\in [j-1]\setminus \{j\}$ arrives in $Y_j$'s window $[\theta_j-1-t,\theta_j)$ and since they must be outside $Y_i$'s window this reduces to the event that not all $j-2$  arrival times fall in the set $[-1,\theta_j-t-1)\cup (\theta_i,1]$. Therefore,
\begin{align*}
&\PP(\ALG(t)=Y_i \wedge J=j \wedge \theta_i\in (t,1] \wedge t\leq \theta_j<x)\\
&=\frac12\int_t^1\PP(\theta_j\in [t,x))\PP(\forall l\in [j-1]\setminus \{i\}, \theta_l\not \in [x-t-1,x] \text{ and not all in } [-1,\theta_j-t-1)\cup (x,1]) dx\\
&=\frac12\int_t^1\int_t^x \left(\left(\frac{1-t}{2}\right)^{j-2}-\left(\frac{1-x+y-t}{2}\right)^{j-2}\right)dydx\\
&=\frac{1}{2^j}\int_t^1 (1-t)^{j-2}(x-t) - \frac{(1-t)^{j-1}-(1-x)^{j-1}}{j-1}dx =\frac{(1-t)^{j}}{2^{j+1}} - \frac{(1-t)^j}{2^jj}.
\end{align*}

Putting all cases together we get that for all $1\leq i\leq k-1$ and $i+1\leq j\leq k-1$,
\begin{align}
\PP(\ALG(t)=Y_i \wedge J=j) &=  \frac{1}{2^j}\cdot\left( \frac{1-(1-t)^{j-1}(1+t)}{j-1}+(1-t)^{j-1}(1+t)- \frac{(1-t)^j}{j}\right).\label{akt}
\end{align}

Now let us consider the case $J=k$. In this case we only need to impose that all $k-2$ arrival times $\theta_l, l\in [k-1]\setminus\{i\}$ fall outside the window of $Y_i$. No matter who  the couple of $k$ is, the previous condition implies that $k$ is inside the window of $Y_i$. Therefore,
\begin{align}
\PP(\ALG(t)=Y_i \wedge J=k)&=\frac{1}{2}\int_{0}^t\left(\frac{1-x}{2}\right)^{k-2}dx + \left(\frac{1-t}{2}\right)^{k-1}\notag\\
&=\frac{1}{2^{k-1}}\left( \frac{1-(1-t)^{k-1}}{k-1} + (1-t)^{k-1}\right). \label{bkt}
\end{align}

By combining \eqref{akt} and \eqref{bkt} together we finish the proof of the first equality of this lemma, \eqref{dist1}.

We now compute the probability that $\ALG(t)$ selects one of the top $k$ values $Y_1,\dots, Y_k$. Observe first that the algorithm never selects a value $Y_j$ with $j>k$. This is because the window of $Y_j$ always contains $\theta_k$ or $C(\theta_k) = \theta_{k'}$, and both $Y_k$ and $Y_{k'}$ are larger than $Y_j$. By this observation, $\PP(\ALG(t)\geq Y_k)$ equals the probability that the algorithm stops by selecting something.

Let $M$ be the random variable denoting the index of the largest value arriving after 0, i.e. $\theta_M\ge 0$ and $\theta_i<0$, for all $i<M$. By Lemma \ref{lem:MWLemma}, $\ALG(t)$ stops by selecting a value if and only if $Y_M$ is larger than every value arriving in its window $[\max(\theta_M-1-t,1),\theta_M))$. Observe that if $M=k$, this is impossible because the couple of $k$ is always in that interval. Therefore,
\begin{align*}
\PP(\ALG(t)\geq Y_k)&=\sum_{j=1}^k \PP(\ALG(t)\geq Y_k, M=j)\\
&=\sum_{j=1}^k \frac12 \int_{0}^1\PP(\forall i\in [j-1], \theta_i\in [-1, x-1-t)dx.\\
&=\frac{1}{2} + \sum_{j=2}^{k-1}\frac{1}{2}\int_t^1 \left(\frac{x-t}{2}\right)^{j-1}dx = \frac{1}{2}+\sum_{j=2}^{k-1}\frac{(1-t)^j}{2^j j}.\qedhere
\end{align*}
\end{proof}     

Algorithm $\ALG(t)$ has a very poor performance for $k=2$. For example, in the instance $\{(1,1-\varepsilon),(\varepsilon,0)\}$, we have $Y_1=1$, $Y_2=1-\varepsilon$ and $k=2$. By Lemma \ref{lem:dist_algt},
$\PP(\ALG(t)\geq Y_2)= 1/2$. Therefore $\E(\ALG(t))\leq 1/2$, while $\E(\max \DD)\geq 1-\varepsilon$.

We will now  study the behaviour of $\ALG(t)$ when  $k$ is large. For that, define
\begin{align*}
R_k(i,t)&=\frac{\PP(\ALG(t)\geq Y_i)}{\PP(\max \DD\geq Y_i)}=\frac{1}{1-(1/2)^i}\sum_{l=1}^i \sum_{j=l+1}^{k-1}b(k,t)+ \frac{a(j,t)}{2^j}.\\
R_k(t)&=\PP(\ALG(t) \text{ stops}) = \frac{1}{2} + \sum_{j=2}^{k-1}\frac{(1-t)^j}{2^j j}.\intertext{ Consider also their limits as $k\to \infty$, which since $kb(k,t) \to 0$ can be computed as } 
R_\infty(i,t)&=\frac{1}{1-(1/2)^i}\sum_{l=1}^i \sum_{j=l+1}^{\infty}\frac{a(j,t)}{2^j}.\\
R_\infty(t)&=\frac{1}{2} + \sum_{j=2}^{\infty}\frac{(1-t)^j}{2^j j} = \frac{t}{2}+\ln\left(\frac{2}{1+t}\right).
\end{align*}

Recall that the competitiveness of $\ALG(t)$ is at least $\min_{i\geq 1} R_\infty(i,t)$, and we are looking for the value $t$ that maximizes this minimum. 
\begin{figure}[t]
\centering
\includegraphics{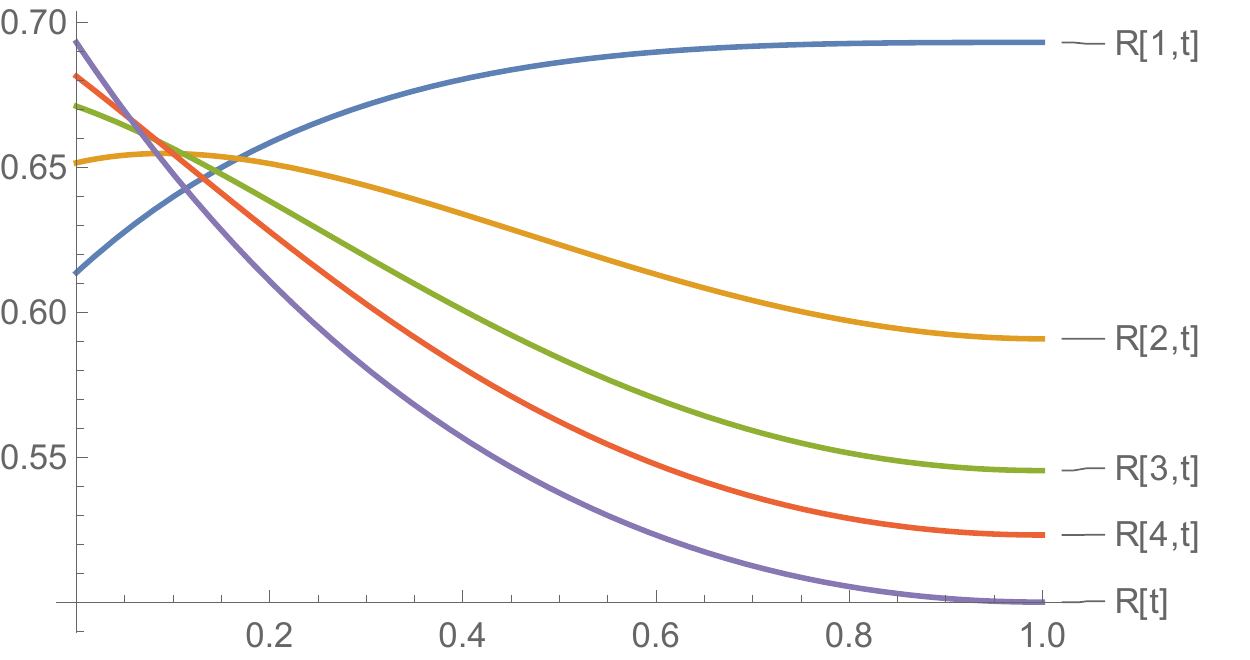}
\caption{Auxiliary functions to analyze $\ALG(t)$, as $k$ tends to $\infty$.}
\label{grafico2}
\end{figure}
In Figure \ref{grafico2} we plot a few of these $R_\infty(i,t)$, for small values of $i$, together with function $R_\infty(t)$. From the plot we observe that the $t^*$ maximizing the minima of all the curves satisfies $R_\infty(1,t^*)=R_\infty(t^*)$. By simplifying the sums, we obtain that $t^*\approx 0.1128904$ is the only root of:
\[R_\infty(1,t)-R_\infty(t)=2-(5 t)/2-(3+t) \ln 2 +(4+t) \ln (1+t).\]
To formally prove that $R_\infty(t^*)\approx 0.6426317$ is the sought guarantee we need the following lemma.
\begin{lem}\label{lem:algt} Let $a(j,t)$ be defined as in Lemma \ref{lem:dist_algt}. For all $j\geq 3$, we have $a(j,t^*)\geq a(j+1,t^*)$
\end{lem}
\begin{proof}
Below, we tabulate a few values of $a(j,t^*)$ with $t^*\approx 0.1128904$, and we observe that  $a(j,t^*)\geq a(j+1,t^*)$ for all $j\in \{3,\dots, 9\}$:
\begin{center}
\begin{tabular}{|c|c|c|c|c|c|c|c|c|}
\hline 
$j$ & 3 & 4 & 5 & 6 & 7 & 8&9&10 \\ 
\hline 
$a(j,t^*)$ &0.705194 &0.696462 &0.65704 &0.607906 &0.556898&0.50734&0.460684 &0.417513 \\ 
\hline 
\end{tabular} 
\end{center}

Let us prove that the inequality also holds for $j\geq 10$. By rearranging terms, $a(j,t)\geq a(j+1,t)$ is equivalent to
\begin{align*}
\frac{1}{(1-t)^{j-1}}\left(\frac{1}{j-1}-\frac{1}{j}\right)&\geq -t(1+t)-\frac{t(1-t)}{j}-\frac{(1-t)^2}{j+1}+\frac{1+t}{j-1}.
\end{align*}
But note that the right hand side, evaluated at $t^*\geq 1/9$ is:
\begin{align*}
-t^*(1+t^*)-\frac{t^*(1-t^*)}{j}-\frac{(1-t^*)^2}{j+1}+\frac{1+t^*}{j-1}\leq (1+t^*)\left(\frac{1}{j-1}-t^*\right)\leq (1+t^*)\left(\frac{1}{j-1}-\frac{1}{9}\right),
\end{align*}
which is nonpositive for all $j\geq 10$. Then, it is smaller than the left hand side which is positive.
\end{proof}

\LargekThm*

\begin{proof}

We will prove that for all $i\geq 1$, $R_\infty(i,t^*)\geq R_\infty(t^*)$. Denote for all $i\geq 1$, $P_i = \lim_{k\to\infty} \PP(\ALG(t^*)=Y_i)=\sum_{j=i+1}^\infty a(j,t^*)/2^j$. By choice of $t^*$, $R_\infty(1,t^*)=R_\infty(t^*)\approx 0.6426317$. Furthermore we we can numerically evaluate $R_\infty(2,t^*)=\frac{P_1+P_2}{1-1/4}\approx 0.6547\geq R_\infty(t^*)$
and $R_\infty(3,t^*)=\frac{P_1+P_2+P_3}{1-1/8}\approx 0.654331 \geq R_\infty(t^*)$.

To finish the proof, we will show that for all $i\geq 4$, $R_\infty(i,t^*)\geq R_\infty(i+1,t^*)$, and therefore, for all $i\geq 4$, $R_\infty(i,t^*)\geq \lim_{j\to \infty} R_\infty(j,t^*)=R_\infty(t^*)$. For this, we will also need the inequality $P_1\approx 0.3213158\geq 0.304065\approx 8P_4$. By Lemma \ref{lem:algt}, for all $i\geq 2$,
\begin{align*}
P_{i}=\sum_{j=i+1}^\infty\frac{a(j,t^*)}{2^j} &\geq \sum_{j=i+1}^\infty\frac{a(j+1,t^*)}{2^j} =2P_{i+1}.\shortintertext{ Therefore, }
\sum_{\ell=1}^{i}P_\ell \geq 8P_4 + \sum_{\ell=2}^iP_{\ell} &\geq 2^iP_{i+1} + P_{i+1}\sum_{\ell=2}^{i}2^{i+1-\ell}=P_{i+1} (2^{i+1}-2).\shortintertext{We conclude that,}
R(i+1,t^*)=\frac{P_{i+1}+\sum_{\ell=1}^{i}P_\ell}{1-1/2^{i+1}} &\leq \frac{\displaystyle\left(1+\frac{1}{2^{i+1}-2}\right)\sum_{\ell=2}^i P_l }{\displaystyle\frac{2^{i+1}-1}{2^{i+1}}}=\frac{\displaystyle\sum_{\ell=2}^i P_l }{\displaystyle\frac{2^{i+1}-2}{2^{i+1}}} = R_\infty(i,t^*).\qedhere
\end{align*}
\end{proof}

To conclude this section we observe that the guarantee obtained in the $k=\infty$ case is still useful for moderately high values of $k$. Indeed, note that for all $j\geq 1$, and all $t\in (0,1)$,
\begin{align*}
|a(j,t)|\leq 2, |b(k,t)|\leq \frac{2}{2^{k-1}}.
\end{align*}
Therefore, for all $1\leq i\leq k-1$ and all $k\geq 3$.
$$
\left|\PP(\ALG(t)=Y_i)-\sum_{j=i+1}^{\infty}\ \frac{a(j,t)}{2^j}\right| =\left|b(k,t)-\sum_{j=k}^\infty \frac{a(j,t)}{2^j}\right|\leq \frac{2}{2^{k-1}} + 2\sum_{j=k}^\infty \frac{1}{2^j}=\frac{4}{2^{k-1}}=\frac{1}{2^{k-3}}.
$$
Therefore, 
$$
|R_k(i,t)-R_\infty(i,t)| =  \frac{1}{1-1/2^k}\left|\sum_{\ell=1}^i\left(\PP(\ALG(t)=Y_\ell) - \sum_{j=\ell+1}^{\infty} \frac{a(j,t)}{2^j}\right)\right|\leq \frac{2i}{2^{k-3}} \leq \frac{16k}{2^{k}}.
$$
We also have that
\begin{align*}
|R_k(t)-R_\infty(t)|&= \left|\sum_{j=k}^\infty \frac{(1-t)^j}{2^j j} \right| \leq \frac{1}{2^{k-1}}\leq \frac{16k}{2^{k}}.
\end{align*}

From here,
\begin{align*}
\min(R_k(t^*), \min_{1\leq i\leq k-1}R_k(i,t^*)) &\geq  
 \min(R_\infty(t^*), \min_{1\leq i\leq k-1}R_k(i,t^*))-\frac{16k}{2^{k}}=  R_\infty(t^*) - \frac{16k}{2^{k}}.
\end{align*}

Thus the guarantee of $\ALG(t)$ for not so large values of $k$ is already very close to $R_\infty(t^*)\approx 0.6426317$. For example, the guarantee of $\ALG(t)$ for $k\geq 20$ is at least $R_\infty(t^*)-0.00031$.


\section*{Acknowledgments}
This work was partially supported by CONICYT under grants CONICYT-PFCHA/Doctorado Nacional/2018-21180347, CONICYT-PFCHA/Magister Nacional/2018-22181138, Conicyt-Fon\-de\-cyt 1181180, Conicyt-Fondecyt 1190043 and PIA AFB-170001, and by an Amazon Research Award.

\end{document}